\documentclass[a4paper, USenglish]{article}
\pdfoutput=1 

\usepackage[utf8]{inputenc}
\usepackage{amsmath, amssymb}

\usepackage{mathtools}

\usepackage[utf8]{inputenc}
\usepackage[nocompress]{cite}
\usepackage{fullpage}
\usepackage{hyperref}
\usepackage{amsmath, amssymb}
\usepackage{graphicx}
\usepackage{xcolor}

\usepackage{amsthm}
\usepackage{mathtools}

\newtheorem{theorem}{Theorem}

\newtheorem{lemma}{Lemma}

\newtheorem{property}{Property}

\newcommand{\lm}{\ensuremath{\textsf{lm}}}

\newcommand{\lep}[1]{\ensuremath{\ell_{#1}}}
\newcommand{\rep}[1]{\ensuremath{r_{#1}}}
\newcommand{\len}{\ensuremath{\textsf{len}}}
\newcommand{\pred}{\ensuremath{\textsf{pred}}}

\newcommand{\lmm}{\ensuremath{\textsf{lmco}}}
\newcommand{\coocc}{\ensuremath{\textsf{co}}}

\newcommand{\A}{\texttt{A}}
\newcommand{\B}{\texttt{B}}
\newcommand{\C}{\texttt{C}}
\newcommand{\D}{\texttt{-}}

\newcommand{\s}{\;}
\newcommand{\garb}{\texttt{\$}\cdots\texttt{\$}}

\newcommand{\m}{\mu}

\title{The Complexity of the Co-Occurrence Problem}
\author{Philip Bille \and Inge Li G{\o}rtz \and Tord Stordalen}
\date{}

\begin{document}
\maketitle
\begin{centering}
\large 
Technical University of Denmark, DTU Compute, Kgs.~Lyngby, Denmark

\texttt{\{phbi,inge,tjost\}@dtu.dk}

\end{centering}

\begin{abstract}
    Let $S$ be a string of length $n$ over an alphabet $\Sigma$ and let $Q$ be a subset of $\Sigma$ of size $q \geq 2$. The \emph{co-occurrence problem} is to construct a compact data structure that supports the following query: given an integer $w$ return the number of length-$w$ substrings of $S$ that contain each character of $Q$ at least once. This is a natural string problem with applications to, e.g., data mining, natural language processing, and DNA analysis. The state of the art is an $O(\sqrt{nq})$ space data structure that --- with some minor additions --- supports queries in $O(\log\log n)$ time~[CPM~2021]. 

    Our contributions are as follows. Firstly, we analyze the problem in terms of a new, natural parameter $d$, giving a simple data structure that uses $O(d)$ space and supports queries in $O(\log\log n)$ time. The preprocessing algorithm does a single pass over $S$, runs in expected $O(n)$ time, and uses $O(d + q)$ space in addition to the input. Furthermore, we show that $O(d)$ space is optimal and that $O(\log\log n)$-time queries are optimal given optimal space. Secondly, we bound $d = O(\sqrt{nq})$, giving clean bounds in terms of $n$ and $q$ that match the state of the art.  Furthermore, we prove that $\Omega(\sqrt{nq})$ bits of space is necessary in the worst case, meaning that the $O(\sqrt{nq})$ upper bound is tight to within polylogarithmic factors. All of our results are based on simple and intuitive combinatorial ideas that simplify the state of the art. 
    
\end{abstract}

\section{Introduction}
We consider the \emph{co-occurrence problem} which is defined as follows. Let $S$ be a string of length $n$ over an alphabet $\Sigma$ and let $Q$ be a subset of $\Sigma$ of size $q \geq 2$. For two integers $i$ and $j$ where $1 \leq i \leq j \leq n$, let  $[i,j]$ denote the discrete interval $\{i,i+1,\ldots,j\}$, and let $S[i,j]$ denote the substring of $S$ starting at $S[i]$ and ending at $S[j]$. The interval $[i,j]$ is a \emph{co-occurrence} of $Q$ in $S$ if $S[i,j]$ contains each character in $Q$ at least once. The goal is to preprocess $S$ and $Q$ into a  data structure that supports the query
\begin{itemize}
    \item $\coocc_{S,Q}(w)$: return the number of co-occurrences of $Q$ in $S$ that have length $w$, i.e., the number of length-$w$ substrings of $S$ that contain each character in $Q$ at least once.
\end{itemize}
For example, let $\Sigma = \{\A,\B,\C,\D\}$, $Q = \{\A,\B,\C\}$ and 
\[
    S \; = \; 
    \underset{1}{\D}\;\;
    \underset{2}{\D}\;\;
    \underset{3}{\D}\;\;
    \underset{4}{\D}\;\;
    \underset{5}{\B}\;\;
    \underset{6}{\C}\;\;
    \underset{7}{\D}\;\;
    \underset{8}{\A}\;\;
    \underset{9}{\C}\;\;
    \underset{10}{\C}\;\;
    \underset{11}{\B}\;\;
    \underset{12}{\D}\;\;
    \underset{13}{\D}\;\;.
\]
Then 
\begin{itemize}
    \item $\coocc_{S,Q}(3) = 0$, because no length-three substring contains all three characters \A, \B, and \C.
    \item $\coocc_{S,Q}(4) = 2$, because both $[5,8]$ and $[8,11]$ are co-occurrences of $Q$. 
    \item  $\coocc_{S,Q}(8) = 6$, because all six of the length-eight substrings of $S$  are co-occurrences of $Q$. 
\end{itemize} 
Note that only sublinear-space data structures are interesting. With linear space we can simply precompute the answer to $\coocc_{S,Q}(i)$ for each $i \in [0,n]$ and support queries in constant time.   

This is a natural string problem with applications to, e.g., data mining, and a large amount of work has gone towards related problems such as finding frequent items in streams~\cite{DLM2002,GDDL+2003,KSP2003,LCK2014} and finding frequent sets of items in streams~\cite{AH2018,CL2006,DP2013,LL2009,LCWC2005,MTZ2008,YYLL+2015}. Furthermore, it is similar to certain string problems, such as episode matching~\cite{DFGG+1997} where the goal is to determine all the substrings of $S$ that occur a certain number of times within a given distance from each other. Whereas previous work is mostly concerned with identifying frequent patterns either in the whole string or in a sliding window of fixed length, Sobel, Bertram, Ding, Nargesian and Gildea~\cite{SBDN+2021} introduced the problem of studying a given pattern across all window lengths (i.e., determining $\coocc_{S,Q}(i)$ for all $i$). They motivate the problem by listing potential applications such as training models for natural language processing (short and long co-occurrences of a set of words tend to represent respectively syntactic and semantic information), automatically organizing the memory of a computer program for good cache behaviour (variables that are used close to each other should be near each other in memory), and analyzing DNA sequences (co-occurrences of nucleotides in DNA provide insight into the evolution of viruses). See~\cite{SBDN+2021} for a more detailed discussion of these applications.

Our work is inspired by~\cite{SBDN+2021}. They do not consider fast, individual queries, but instead they give an $O(\sqrt{nq})$ space data structure from which they can determine  $\coocc_{S,Q}(i)$ for each $i=1,\ldots,n$ in $O(n)$ time. Supporting fast queries is a natural extension to their problem, and we note that their solution can be extended to support individual queries in $O(\log\log n)$ time using the techniques presented below. 

A key component of our result is a solution to the following simplified problem. A co-occurrence $[i,j]$ is \emph{left-minimal} if $[i+1,j]$ is not a co-occurrence. The \emph{left-minimal co-occurrence problem} is to preprocess $S$ and $Q$ into a data structure that supports the query 
\begin{itemize}
    \item $\lmm_{S,Q}(w)$: return the number of left-minimal co-occurrences of $Q$ in $S$ that have length $w$. 
\end{itemize}
We first solve this more restricted problem, and then we solve the co-occurrence problem by a reduction to the left-minimal co-occurrence problem. To our knowledge this problem has not been studied before. 

\subsection{Our Results}
Our two main contributions are as follows. Firstly, we give an upper bound that matches and simplifies the state of the art. Secondly, we provide lower bounds that show that our solution has optimal space, and that our query time is optimal for optimal-space data structures. As in previous work, all our results work on the word RAM model with logarithmic word size.

To do so we use the following parametrization. Let $\delta_{S,Q}$ be the difference encoding of the sequence $\lmm_{S,Q}(1),\ldots,\lmm_{S,Q}(n)$. That is, $\delta_{S,Q}(i) = \lmm_{S,Q}(i) - \lmm_{S,Q}(i-1)$ for each $i \in [2,n]$ (note that $\lmm(1) = 0$ since $|Q| \geq 2$).  Let $Z_{S,Q} = \{i \in [2,n] \mid \delta(i) \neq 0\}$ and let $d_{S,Q} = |Z_{S,Q}|$. For the remainder of the paper we will omit the subscript on $\lmm$, $\coocc$, $Z$, and $d$ whenever $S$ and $Q$ are clear from the context. Note that $d$ is a parameter of the problem since it is determined exclusively by the input $S$ and $Q$. We prove the following theorem. 

\begin{theorem}{\label{thm:overall_results}}
Let $S$ be a string of length $n$ over an alphabet $\Sigma$, let $Q$ be a subset of $\Sigma$ of size $q \geq 2$, and let $d$ be defined as above.
\begin{itemize}
    \item[(a)] There is an $O(d)$ space data structure that supports both $\lmm_{S,Q}$- and $\coocc_{S,Q}$-queries in $O(\log\log n)$ time. The preprocessing algorithm does a single pass over $S$, runs in expected $O(n)$ time and uses $O(d + q)$ space in addition to the input. 
    
    \item[(b)] Any data structure supporting either $\lmm_{S,Q}$- or $\coocc_{S,Q}$-queries needs $\Omega(d)$ space in the worst case, and any $d\log^{O(1)}d$ space data structure cannot support queries faster than $\Omega(\log\log n)$ time.
    
    \item[(c)]  The parameter $d$ is bounded by $O(\sqrt{nq})$, and any data structure supporting either $\lmm_{S,Q}$- or $\coocc_{S,Q}$-queries needs $\Omega(\sqrt{nq})$ \emph{bits} of space in the worst case. 
    \end{itemize}
\end{theorem}

Theorem~\ref{thm:overall_results}(a) and~\ref{thm:overall_results}(b) together prove that our data structure has optimal space, and that with optimal space we cannot hope to support queries faster than $O(\log\log n)$ time. In comparison to the state of the art by Sobel et al.~\cite{SBDN+2021}, Theorem~\ref{thm:overall_results}(c) proves that we match their $O(\sqrt{nq})$ space and $O(\log\log n)$ time solution, and also that the $O(\sqrt{nq})$ space bound is tight to within polylogarithmic factors. All of our results are based on simple and intuitive combinatorial ideas that simplify the state of the art. 

Given a set $X$ of $m$ integers from a universe $U$, the \emph{static predecessor problem} is to represent $X$ such that we can efficiently answer the query $\textsf{predecessor}(x) = \max\{y \in X \mid y \leq x\}$. Tight bounds by P\u{a}tra\c{s}cu and Thorup~\cite{PT2007} imply that $O(\log\log |U|)$-time queries are optimal with $m\log^{O(1)}m$ space when $|U| = m^c$ for any constant $c > 1$. The lower bound on query time in Theorem~\ref{thm:overall_results}(b) follows from the following theorem, which in turn follows from a reduction from the predecessor problem to the (left-minimal) co-occurrence problem. 
\begin{theorem}{\label{thm:time_space_tradeoff}} Let $X \subseteq \{2,\ldots,u\}$ for some $u$ and let $|X| = m$. Let $n$, $q$, and $d$ be the parameters of the (left-minimal) co-occurrence problem as above. Given a data structure that supports $\lmm$- or $\coocc$-queries in $f_t(n,q,d)$ time using $f_s(n,q,d)$ space, we obtain a data structure that supports predecessor queries on $X$ in $O(f_t(2u^2, 2, 8m))$ time using $O(f_s(2u^2, 2, 8m))$ space.
\end{theorem}
In particular, if $f_s(n,q,d) = d\log^{O(1)} d$ then we obtain an $m\log^{O(1)}m$-space predecessor data structure on $X$. If also $u = m^c$, then it follows from the lower bound on predecessor queries that $f_t(2u^2,2,8m) = \Omega(\log\log u)$, which in turn implies that $f_t(n,q,d) = \Omega(\log\log n)$, proving the lower bound in Theorem~\ref{thm:overall_results}(b).

The preprocessing algorithm and the proof of Theorem~\ref{thm:time_space_tradeoff} can be found in Appendices~\ref{app:preprocessing} and~\ref{app:lower_bound_on_time}, respectively.

\subsection{Techniques}
The key technical insights that lead to our results stem mainly from the structure of $\delta$. 

To achieve the upper bound for $\lmm$-queries we use the following very simple data structure. By definition, $\lmm(w) = \sum_{i=2}^w\delta(i)$. Furthermore, by the definition of $Z$ it follows that for any  $w \in [2,n]$ we have that $\lmm(w) = \lmm(w_p)$ where $w_p$ is the predecessor of $w$ in $Z$.  Our data structure is a predecessor structure over the set of key-value pairs $\{(i,\lmm(i)) \mid i \in Z\}$ and answers $\lmm$-queries with a single predecessor query. There are linear space predecessor structures that support queries in $O(\log\log |U|)$ time~\cite{Willard1983}. Here the universe $U$ is $[2,n]$ so we match the $O(d)$ space and $O(\log\log n)$ time bound in Theorem~\ref{thm:overall_results}(a). 

Furthermore, we prove the $O(\sqrt{nq})$ upper bound on space by bounding $d = O(\sqrt{nq})$ using the following idea. In essence, each $\delta(z)$ for $z \in Z$ corresponds to some length-$z$ \emph{minimal co-occurrence}, which is a co-occurrence $[i,j]$ such that neither $[i+1,j]$ nor $[i,j-1]$ are co-occurrences (see below for the full details on $\delta$). We bound the cumulative length of all the minimal co-occurrences to be $O(nq)$; then there are at most $d = O(\sqrt{nq})$ distinct lengths of minimal co-occurrences since $d = \omega(\sqrt{nq})$ implies that the cumulative length of the minimal co-occurrences is at least $1 + \ldots + d = \Omega(d^2) = \omega(nq)$.

To also support $\coocc$-queries and complete the upper bound we give a straight-forward reduction from the co-occurrence problem to the left-minimal co-occurrence problem. We show that by extending the above data structure to also store $\sum_{i=2}^z\lmm(i)$ for each $z \in Z$, we can support $\coocc$-queries with the same bounds as for $\lmm$-queries.

On the lower bounds side, we give all the lower bounds for the left-minimal co-occurrence problem and show that they extend to the co-occurrence problem. To prove the lower bounds we exploit that we can carefully design $\lmm$-instances that result in a particular difference encoding $\delta$ by including minimal co-occurrences of certain lengths and spacing. Our lower bounds on space in Theorem~\ref{thm:overall_results}(b) and~\ref{thm:overall_results}(c) are the results of encoding a given permutation or set in $\delta$, respectively.

Finally, as mentioned above, we prove Theorem~\ref{thm:time_space_tradeoff} (and, by extension, the lower bound on query time in Theorem~\ref{thm:overall_results}(b)) by encoding a given instance of the static predecessor problem in an $\lmm$-instance such that the predecessor of an element $x$ equals $\lmm(x)$.

\section{The Left-Minimal Co-Occurrence Problem}\label{sec:lmm_problem}
\subsection{Main Idea}
Let $[i,j] \subseteq [1,n]$ be a co-occurrence of $Q$ in $S$. Recall that then each character from $Q$ occurs in $S[i,j]$ and that $[i,j]$ is left-minimal if $[i+1,j]$ is not a co-occurrence. The goal is to preprocess $S$ and $Q$ to support the query  $\lmm(w)$, which returns the number of left-minimal co-occurrences of length $w$.

We say that $[i,j]$ is a \emph{minimal co-occurrence} if neither $[i+1,j]$ nor $[i,j-1]$ are co-occurrences and we denote the $\m$ minimal co-occurrences of $Q$ in $S$ by $[\lep{1},\rep{1}],\ldots,[\lep{\m},\rep{\m}]$ where $\rep{1} < \ldots < \rep{\m}$. This ordering is unique since at most one minimal co-occurrence ends at a given index. To simplify the presentation we define $\rep{\m+1} = n+1$. Note that also $\lep{1} < \ldots < \lep{\m}$ due to the following property. 
\begin{property}{\label{prop:minimal_match_order}}
    Let $[a,b]$ and $[a',b']$ be two minimal co-occurrences. Either both $a < a'$ and $b < b'$, or both $a' < a$ and $b' < b$. 
\end{property}
\begin{proof}
    If $a < a'$ and $b \geq b'$ then $[a,b]$ strictly contains another minimal co-occurrence $[a',b']$ and can therefore not be minimal itself. The other cases are analogous.  
    
    \end{proof}
We now show that given all the minimal co-occurrences we can determine all the left-minimal co-occurrences. 

\begin{lemma}{\label{lem:lmm_start}}
    Let $[\lep{1},\rep{1}],\ldots,[\lep{\m},\rep{\m}]$ be the minimal co-occurrences of $Q$ in $S$ where $\rep{1} < \ldots < \rep{\m}$ and let $\rep{\m+1} = n+1$. Then
    \begin{itemize}
        \item[(a)] there are no left-minimal co-occurrences that end before $\rep{1}$, i.e., at an index $k < \rep{1}$, and
        \item[(b)] for each index $k$ where $\rep{i} \leq k < \rep{i+1}$ for some $i$, the left-minimal co-occurrence ending at $k$ starts at $\lep{i}$.     
    \end{itemize}
\end{lemma}
\begin{proof}
    See Fig.~\ref{fig:lmm_start_example} for an illustration of the proof. 
    \begin{itemize}
        \item[(a)] If $[j,k]$ is a left-minimal co-occurrence where $k < \rep{1}$, it must contain some minimal co-occurrence that ends before $\rep{1}$ --- obtainable by shrinking $[j,k]$ maximally --- leading to a contradiction. 
            \item[(b)] Let $\rep{i} \leq k < \rep{i+1}$. Then $[\lep{i}, k]$ is a co-occurrence since it contains $[\lep{i},\rep{i}]$. We show that it is left-minimal by showing that $[\lep{i}+1,k]$ is not a co-occurrence. If it were, it would contain a minimal co-occurrence $[s,t]$ where $\lep{i} < s$ and $t < \rep{i+1}$. By Property~\ref{prop:minimal_match_order}, $\lep{i} < s$ implies that $\rep{i} < t$. However, then $\rep{i} < t < \rep{i+1}$, leading to a contradiction.
        \end{itemize}
        
        \end{proof}

        \begin{figure}
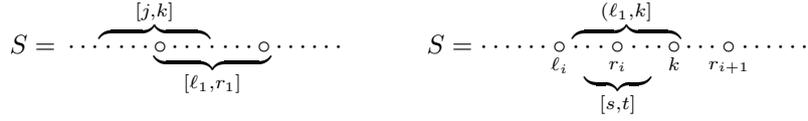

            \newcommand{\ic}[2]{\underset{\mathclap{#1}}{#2}}
            
            \[
                S = 
                \lefteqn{\phantom{\cdots}\overbrace{\phantom{\cdot\cdots\circ\cdots\cdot}}^{[j,k]}}
                \cdots\cdot\cdots\underbrace{\circ\cdots\cdot\cdots\circ}_{[\lep{1},\rep{1}]}
                \cdots\cdots
                \qquad\quad
                S = \cdots\cdots 
                \ic{\lep{i}}{\circ} \overbrace{\cdot\underbrace{\cdot\cdot \ic{\rep{i}}{\circ}
                 \cdot\cdot}_{[s,t]}\cdot\,\ic{k}{\circ}}^{(\lep{1},k]}
                 \cdots \ic{\rep{i+1}}{\circ} \cdots\cdots
            \]

            \caption{\emph{Left (Lemma~\ref{lem:lmm_start}(a)):} Any left-minimal co-occurrence $[j,k]$ must contain a minimal co-occurrence ending at or before $k$. If $k < \rep{1}$ this contradicts that $\rep{1}$ is the smallest endpoint of a minimal co-occurrence. \emph{Right (Lemma~\ref{lem:lmm_start}(b)):} $[\lep{i},k]$ is a co-occurrence because it contains $[\lep{i},\rep{i}]$. However, $[\lep{i}+1,k]$ is not a co-occurrence; if it were it would contain a minimal co-occurrence $[s,t]$ that ends between $\rep{i}$ and $k$, leading to a contradiction since $\rep{i} < t < \rep{i+1}$.}
            \label{fig:lmm_start_example}
        \end{figure}
        
Let $\len(i,j) = j - i + 1$ denote the length of the interval $[i,j]$. Lemma~\ref{lem:lmm_start} implies that each minimal co-occurrence $[\lep{i},\rep{i}]$ gives rise to one additional left-minimal co-occurrence of length $k$ for $k = \len(\lep{i},\rep{i}),\ldots,\len(\lep{i},\rep{i+1})-1$. Also, note that each left-minimal co-occurrence is determined by a minimal co-occurrence in this manner. Therefore, $\lmm(w)$ equals the number of minimal co-occurrences $[\lep{i},\rep{i}]$ where $\len(\lep{i},\rep{i}) \leq w < \len(\lep{i},\rep{i+1})$. Recall that $\delta(i) = \lmm(i) - \lmm(i-1)$ for $i\in [2,n]$. Since $\lmm(1) = 0$ (because $|Q| \geq 2$) we have $\lmm(w) = \sum_{i=2}^w \delta(i)$. It follows that
\[
    \delta(w) = \sum_{i=1}^{\m} \begin{cases}
            1 & \text{if } \len(\lep{i},\rep{i}) = w \\
            -1 & \text{if } \len(\lep{i}, \rep{i+1}) = w\\
            0 & \text{otherwise } 
        \end{cases}
\]
since the contribution of each $[\lep{i},\rep{i}]$ to the sum $\sum_{i=2}^w\delta(i)$    
is one if $\len(\lep{i},\rep{i}) \leq w < \len(\lep{i},\rep{i+1})$ and zero otherwise. We say that $[\lep{i},\rep{i}]$ \emph{contributes} plus one and minus one to $\delta(\len(\lep{i},\rep{i}))$ and $\delta(\len(\lep{i}, \rep{i+1}))$, respectively.

However, note that only the non-zero $\delta(\cdot)$-entries affect the result of $\lmm$-queries. Denote the elements of $Z$ by $z_1 < z_2 < \ldots <  z_d$ and define $\pred(w)$ such that $z_{\pred(w)}$ is the predecessor of $w$ in $Z$, or $\pred(w) = 0$ if $w$ has no predecessor. We get the following lemma. 
\begin{lemma}{\label{lem:sparse_lmm_evaluation}}
For any $w \in [z_1,n]$ we have that 
\[
    \lmm(w) = \sum_{i=1}^{\pred(w)} \delta(z_i) = \lmm(z_{\pred(w)}).
\]
For $w \in [0,z_1)$, $w$ has no predecessor in $Z$ and $\lmm(w) = 0$. 
\end{lemma}
\begin{proof}
The proof follows from the fact that $\lmm(w) = \sum_{i=2}^w \delta(i)$ and $\delta(i) = 0$ for each $i \not\in Z$.

\end{proof}

\subsection{Data Structure}
The contents of the data structure are as follows. Store the linear space predecessor structure from~\cite{Willard1983} over the set $Z$, and for each key $z_i \in Z$ store the data $\lmm(z_i)$. To answer $\lmm(w)$, find the predecessor $z_{\pred(w)}$ of $w$ and return $\lmm(z_{\pred(w)})$.  Return $0$ if $w$ has no predecessor.

The correctness of the query follows from Lemma~\ref{lem:sparse_lmm_evaluation}. The query time is $O(\log\log |U|)$~\cite{Willard1983} which is $O(\log\log n)$ since the universe is $[2,n]$, i.e., the domain of $\delta$. The predecessor structure uses $O(|Z|) = O(d)$ space, which we now show is $O(\sqrt{nq})$. We begin by bounding the cumulative length of the minimal co-occurrences. 

\begin{lemma}{\label{lem:minimal_match_crossing_index}}
    Let $[\lep{1},\rep{1}],\ldots,[\lep{\m},\rep{\m}]$ be the minimal co-occurrences of $Q$ in $S$. Then 
    \[
        \sum_{i=1}^{\m} \len(\lep{i},\rep{i}) = O(nq).
    \] 
\end{lemma}
\begin{proof}
    We prove that for each $k \in [1,n]$ there are at most $q$ minimal co-occurrences $[\lep{i},\rep{i}]$ where $k \in [\lep{i},\rep{i}]$; the statement in the lemma follows directly. Suppose that there are $q' > q$ minimal co-occurrences $[s_1,t_1],\ldots,[s_{q'},t_{q'}]$ that contain $k$ and let $t_1 < \ldots < t_{q'}$. By Property~\ref{prop:minimal_match_order}, and because each minimal occurrence contains $k$, we have
    \[s_1 < \ldots < s_{q'} \leq \; k  \; \leq t_1 < \ldots < t_{q'}\]
     Furthermore, for each $s_i$ we have that $S[s_i] = p$ for some $p \in Q$; otherwise $[s_i+1,t_i]$ would be a co-occurrence and $[s_i,t_i]$ would not be minimal. Since $q' > q$ there is some $p \in Q$ that occurs twice as the first character, i.e., such that $S[s_i] = S[s_j] = p$ for some $i < j$. However, then $[s_i+1,t_i]$ is a co-occurrence because it still contains $S[s_j] = p$, contradicting that $[s_i,t_i]$ is minimal.
     
\end{proof}

By the definition of $\delta$, we have that $\delta(k) \neq 0$ only if there is some minimal co-occurrence $[\lep{i},\rep{i}]$ such that either $\len(\lep{i},\rep{i}) = k$ or $\len(\lep{i},\rep{i+1}) = k$. Using this fact in conjunction with Lemma~\ref{lem:minimal_match_crossing_index} we bound the sum of the elements in $Z$.
\begin{align*}
    \sum_{z \in Z} z = \sum_{k \text{ where } \delta(k) \neq 0} k \quad
    & \leq \sum_{i=1}^{\m} \len(\lep{i},\rep{i}) + \len(\lep{i},\rep{i+1})\\ 
    &= \sum_{i=1}^{\m} \len(\lep{i},\rep{i}) + \Big( \len(\lep{i},\rep{i}) + \len(\rep{i} + 1, \rep{i+1})\Big)\\
    &= \sum_{i=1}^{\m} 2\cdot\len(\lep{i},\rep{i}) + \sum_{i=1}^{\m} \len(\rep{i}+1,\rep{i+1})\\
    &= O(nq) + O(n)
\end{align*}
Since the sum over $Z$ is at most $O(nq)$ we must have $d = O(\sqrt{nq})$, because with $d = \omega(\sqrt{nq})$ distinct elements in $Z$ we have 
\[
    \sum_{z \in Z} z \geq 1 + 2 + \ldots + d = \Omega(d^2) = \omega(nq).    
\]

\section{The Co-Occurrence Problem}\label{sec:coocc_problem}
Recall that $\coocc(w)$ is the number of co-occurrences of length $w$, as opposed to the number of left-minimal co-occurrences of length $w$. That is, $\coocc(w)$ counts the number of co-occurrences among the intervals $[1,w], [2,w+1],\ldots,[n-w+1,n]$. We reduce the co-occurrence problem to the left-minimal co-occurrence problem as follows. 

\begin{lemma} \label{lem:coocc_eval_lmm}

Let $S$ be a string over an alphabet $\Sigma$, let $Q \subseteq \Sigma$ and let $\lmm$ be defined as above. Then
\[
  \coocc(w) = \left(\sum_{i=2}^w \lmm(i)\right) \;-\; \max(w - \rep{1},0) .
\]

\end{lemma}
\begin{proof}
For any index $k \geq w$, the length-$w$ interval $[k-w+1, k]$ ending at $k$ is a co-occurrence if and only if the length of the left-minimal co-occurrence ending at index $k$ is at most $w$. The sum $\sum_{i=2}^w\lmm(i)$ counts the number of indices $j \in [1,n]$ such that the left-minimal co-occurrence ending at $j$ has length at most $w$. However, this also includes the left-minimal co-occurrences that end at any index $j \in [\rep{1},w-1]$. While all of these have length at most $w-1$, none of the length-$w$ intervals that end in the range $[\rep{1},w-1]$ correspond to substrings of $S$, so they are not co-occurrences. Therefore, the sum $\sum_{i=2}^w\lmm(i)$ overestimates $\coocc(w)$ by $w - \rep{1}$ if $\rep{1} < w$ and by $0$ otherwise.

\end{proof}

We show how to represent the sequence $\sum_{i=2}^2\lmm(i),\ldots,\sum_{i=2}^n\lmm(i)$ compactly, in a similar way to what we did for $\lmm$-queries. Recall that the elements of $Z$ are denoted by $z_1 < \ldots < z_d$, that $z_{\pred(x)}$ is the predecessor of $x$ in $Z$, and that $\pred(x) = 0$ if $x$ has no predecessor. Then, for any $w \geq 2$ we get that
\begin{equation}{\label{eq:lmm_sum_rewrite}}
     \begin{split}
     \sum_{i=2}^w\lmm(i) &= \sum_{i=2}^w\sum_{j=1}^{\pred(i)} \delta(z_j)\\
     &= \sum_{k=1}^{\pred(w)}\delta(z_k)(w - z_k + 1)\\
     &= (w+1)\underbrace{\sum_{k=1}^{\pred(w)}\delta(z_k)}_{\lmm(w)} - \sum_{k=1}^{\pred(w)}z_k\delta(z_k)
     \end{split}
\end{equation}
The first step follows by Lemma~\ref{lem:sparse_lmm_evaluation} and the second step follows because $\delta(z_k)$ occurs in $\sum_{j=1}^{\pred(i)}\delta(z_j)$ for each of the $w - z_k +1$ choices of $i \in [z_k,w]$.

To also support $\coocc$-queries we extend our data structure from before as follows. For each $z_k$ in the predecessor structure we store $\sum_{i=1}^k z_i\delta(z_i)$ in addition to $\lmm(z_k)$. We also store $\rep{1}$. Using Lemma~\ref{lem:coocc_eval_lmm} and Equation~\ref{eq:lmm_sum_rewrite} we can then answer $\coocc$-queries with a single predecessor query and a constant amount of extra work, taking $O(\log\log n)$ time. The space remains $O(d) = O(\sqrt{nq})$.  This completes the proof of Theorem~\ref{thm:overall_results}(a), as well as the upper bound on space from Theorem~\ref{thm:overall_results}(c). 

\section{Lower Bounds}\label{sec:lower_bounds}
In this section we show lower bounds on the space complexity of data structures that support $\lmm$- or $\coocc$-queries. In Section~\ref{sec:increment_gadget} we introduce a gadget that we use in Section~\ref{sec:lower_bound_space_d} to prove that any data structure supporting $\lmm$- or $\coocc$-queries needs $\Omega(d)$ space (we use the same gadget in Appendix~\ref{app:lower_bound_on_time} to prove Theorem~\ref{thm:time_space_tradeoff}). In Section~\ref{sec:lower_bound_space_n_and_q} we prove that any solution to the (left-minimal) co-occurrence problem requires $\Omega(\sqrt{nq})$ bits of space in the worst case.

All the lower bounds are proven by reduction to the left-minimal co-occurrence problem. However, they extend to data structures that support $\coocc$-queries by the following argument. Store  $\rep{1}$ and any data structure that supports $\coocc$ on $S$ and $Q$ in time $t$ per query. Then this data structure supports $\lmm$-queries in $O(t)$ time, because by Lemma~\ref{lem:coocc_eval_lmm} we have that
\begin{align*}
    \coocc & (w)\,- \,\coocc(w-1) \\
    &= \left(\sum_{i=2}^w \lmm(i) \;-\; \max(w-\rep{1},0)\right) - \left(\sum_{i=2}^{w-1} \lmm(i) \;-\; \max(w-1-\rep{1},0)\right)\\
    &= \lmm(w) - \max(w-\rep{1},0) + \max(w-1-\rep{1},0)
\end{align*}

\subsection{The Increment Gadget}\label{sec:increment_gadget}
Let $Q = \{\A, \B\}$ and $U = \{2,\ldots, u\}$. For each $i \in U$ we define the \emph{increment gadget} 
\[
    G_i = \underbrace{\A\s\garb\s\B}_{i}\s\underbrace{\garb}_{u}
\]
where $\garb$ denotes a sequence of characters that are not in $Q$. 

\begin{lemma}{\label{lem:increment_gadget}}
    Let $Q = \{\A, \B\}$,  $U = \{2,\ldots,u\}$, and let $G_i$ be defined as above.
    Furthermore, for some $E = \{e_1,e_2,\ldots,e_m\} \subseteq U$ let $S$ be the concatenation of $c_1 > 0$ copies of $G_{e_1}$, with $c_2 > 0$ copies of $G_{e_2}$, and so on. That is,
    \[
        S = \underbrace{G_{e_1}\cdots G_{e_1}}_{c_1}\s\ldots\ldots\ldots\s\underbrace{G_{e_m}\cdots G_{e_m}}_{c_m}
    \]
    Then $\delta(e_i) = c_i$ for each $e_i \in E$ and $\delta(e) = 0$ for any $e \in  U\setminus E$. Furthermore, $m \leq d \leq 8m$ and $n \leq 2uC$ where $C = \sum_{i=1}^m c_i$ is the number of gadgets in $S$. 
\end{lemma} 
\begin{proof}
     Firstly, $|G_j| = j + u  \leq 2u$ since $j \in U$, so the combined length of the $C$ gadgets is at most $2uC$. 

     Now we prove that $\delta(e_i) = c_i$ for each $e_i \in E$ and $\delta(e) = 0$ for each $e \in U\setminus E$. Consider two gadgets $G_j$ and $G_k$ that occur next to each other in $S$. 
     \[
         \overbrace{\underbrace{\A\s\garb\s\B}_{j}\s\underbrace{\garb}_{u}}^{G_j}\s\overbrace{\underbrace{\A\s\garb\s\B}_{k}\s\underbrace{\garb}_{u}}^{G_k}
     \]
     Three of the minimal co-occurrences in $S$ occur in these two gadgets. Denote them by $[s_1,t_1]$,$[s_2,t_2]$ and $[s_3,t_3]$. 
     \[
        \lefteqn{\underbrace{\phantom{\A\s\garb\s\B}}_{[s_1,t_1]}}
        \A\s\garb\s
        \lefteqn{\overbrace{\phantom{\B\s\garb\s\A}}^{[s_2,t_2]}}
        \B\s\garb\s
        \underbrace{\A\s\garb\s\B}_{[s_3,t_3]}\s\garb
     \]
     The two first minimal co-occurrences start in $G_j$. They contribute
     \begin{itemize}
         \item plus one to $\delta(x)$ for $x \in \{\len(s_1,t_1),\len(s_2,t_2)\} = \{j,u+2\}$.
         \item minus one to $\delta(x)$ for $x \in \{\len(s_1,t_2),\len(s_2,t_3)\} = \{j + u + 1, k + u + 1\}$. 
     \end{itemize}
     Hence, each occurrence of $G_j$ contributes plus one to $\delta(j)$, and the remaining contributions are to $\delta(x)$ where $x \not\in U$. The argument is similar also for the last gadget in $S$ that has no other gadget following it. For each $e_i \in E$ there are $c_i$ occurrences of $G_{e_i}$ so $\delta(e_i) = c_i$. For each $e \in U\setminus E$ there are no occurrences of $G_e$ so $\delta(e) = 0$. 
    
     Finally, note that each occurrence of $G_jG_k$ at different positions in $S$ contributes to the same four $\delta(\cdot)$-entries. Therefore the number of distinct non-zero $\delta(\cdot)$-entries is linear in the number of distinct pairs $(j,k)$ such that $G_j$ and $G_k$ occur next to each other in $S$. Here we have no more than $2m$ distinct paris since $G_{e_i}$ is followed either by $G_{e_i}$ or $G_{e_{i+1}}$. Each distinct pair contributes to at most four $\delta(\cdot)$-entries so $d \leq 8m$. Finally each $G_{e_i}$ contributes at least to $\delta(e_i)$ so $m \leq d$, concluding the proof.
     
\end{proof}

\subsection{Lower Bound on Space}\label{sec:lower_bound_space_d}
We prove that any data structure supporting $\lmm$-queries needs $\Omega(d)$ space in the worst case. Let $U$ and $Q$ be defined as in the increment gadget and let $P = p_2,\ldots,p_m$ be a sequence of length $m-1$ where each $p_i \in U$ (the first element is named $p_2$ for simplicity). We let $S$ be the concatenation of $p_2$ occurrences of $G_2$, with $p_3$ occurrences of $G_3$, and so on. That is,
\[
    S = \underbrace{G_2\ldots G_2}_{p_2}\s\cdots\cdots\cdots\s\underbrace{G_m\ldots G_m}_{p_m}
\]
Then any data structure supporting $\lmm$ on $S$ and $Q$ is a representation of $P$; by Lemma~\ref{lem:increment_gadget} we have that $\delta(i) = p_i$ for $i \in [2,m]$  and by definition we have $\delta(i) = \lmm(i) - \lmm(i-1)$.

The sequence $P$ can be any one of $(u-1)^{m-1}$ distinct sequences, so any representation of $P$ requires
\[
    \log((u-1)^{m-1}) = (m-1)\log(u-1) = \Omega(m\log u)
\] bits --- or $\Omega(m)$ words --- in the worst case. By Lemma~\ref{lem:increment_gadget} this is $\Omega(d)$.

\subsection{\texorpdfstring{Lower Bound on Space in Terms of $\boldsymbol{n}$ and $\boldsymbol{q}$}{Lower Bound on Space in Terms of n and q}}\label{sec:lower_bound_space_n_and_q}
Here we prove that any data structure supporting $\lmm$ needs $\Omega(\sqrt{nq})$ bits of space in the worst case. 

The main idea is as follows. Given an integer $\alpha$ and some $k \in \{2,\ldots,\alpha\}$, let $V$ be the set of \emph{even} integers from $\{k+1,\ldots, k\alpha\}$, and let $T$ be some subset of $V$. We will construct an instance $S$ and $Q$ where 
\begin{itemize}
    \item the size of $Q$ is $q = k$
    \item the length of $S$ is $n = O(k\alpha^2)$
    \item for each $i \in V$ we have $\delta(i) = 1$ if and only if $i \in T$. 
\end{itemize}
Then, as above, any data structure supporting $\lmm$-queries on $S$ and $Q$ is a representation of $T$ since $\delta(i) = \lmm(i) - \lmm(i-1)$. There are $2^{\Omega(k\alpha)}$ choices for $T$, so any representation of $T$ requires 
\[
    \log 2^{\Omega(k\alpha)} = \Omega(k\alpha) = \Omega(\sqrt{k^2\alpha^2}) = \Omega(\sqrt{nq})
\]
bits in the worst case.

The reduction is as follows. Let $Q = \{\C_1,\ldots,\C_k\}$ and let $\texttt{\$}$ be a character not in $Q$. Assume for now that $|T|$ is a multiple of $k-1$ and partition $T$ arbitrarily into $t = O(\alpha)$ sets $T_1,\ldots,T_t$, each of size $k-1$. Consider $T_j = \{e_1,\ldots,e_{k-1}\}$ where $e_1 < e_2 < \ldots < e_{k-1}$. We encode $T_j$ in the gadget $R_j$ where 
\begin{itemize}
    \item the length of $R_j$ is $3k\alpha$.
    \item $R_j[1,k] = \C_1\C_2\ldots\C_{k}$.
    \item $R_j[i + e_i] = \C_i$ for each $\C_i$ except $\C_k$. This is always possible since $i + e_i < (i+ 1) + e_{i+1}$.
    \item all other characters are \texttt{\$}.
\end{itemize}

\begin{figure}
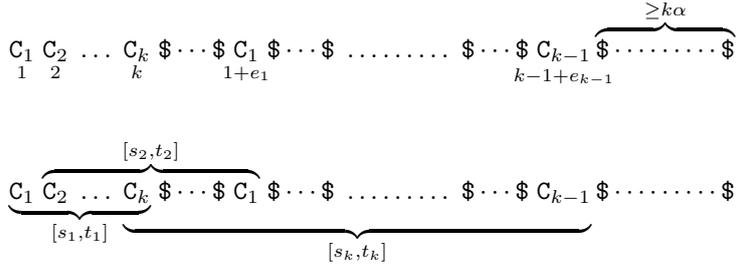

    \[
    \underset{1}{\C_1}\s\underset{2}{\C_2}\s\ldots\s\underset{k}{\C_k}\s\garb\s\underset{\mathclap{1 + e_1}}{\C_1}\s\garb\s\ldots\ldots\ldots\s\garb\s\underset{\mathclap{k-1 + e_{k-1}}}{\C_{k-1}}\overbrace{\texttt{\$}\cdots\cdots\cdots\texttt{\$}}^{\geq k\alpha}    
    \]\\
    \[
        \lefteqn{\underbrace{\phantom{\C_1\s\C_2\s\ldots\s\C_k}}_{[s_{1},t_{1}]}}
        \C_1\s
        \lefteqn{\overbrace{\phantom{\C_2\s\ldots\s\C_k\s\garb\s\C_1}}^{[s_{2},t_{2}]}}
        \C_2\s\ldots\s
        \underbrace{\underset{\phantom{a}}{\C_k}\s\garb\s\C_1\s\garb\s\ldots\ldots\ldots\s\garb\s\C_{k-1}}_{[s_{k},t_{k}]}
        \texttt{\$}\cdots\cdots\cdots\texttt{\$}
    \]
    \caption{\emph{Top:} Shows the layout of the gadget $R_j$, where $\garb$ denotes a sequence of characters not in $Q$. The first $k$ characters are $\C_1\ldots\C_k$. For $i \in [1,k-1]$ there is another occurrence of $\C_i$ at index $i + e_i$. Note that the second occurrence of $\C_1$ occurs before the second occurrence of $\C_2$, and so on. All other characters are $\texttt{\$}$. Since $|R_j| = 3k\alpha$ and $k-1 + e_{k-1} \leq 2k\alpha$, $R_j$ ends with at least $k\alpha$ characters that are not in $Q$. \emph{Bottom:} Shows the $k$ minimal co-occurrences in $R_j$ denoted by $[s_{1},t_{1}],\ldots,[s_{k},t_{k}]$. Each of the $k-3$ minimal co-occurrences that are not depicted start at the first occurrence of some $\C_i$ and ends at the second occurrence of $\C_{i-1}$.}
    \label{fig:nq_gadget_illustration}
\end{figure}

See Fig.~\ref{fig:nq_gadget_illustration} for an illustration both of the layout of $R_j$ and of the minimal co-occurrences contained within it. There are $k$ minimal co-occurrences contained in $R_j$ which we denote by $[s_{1},t_{1}],\ldots,[s_{k},t_{k}]$.  
\begin{itemize}
    \item The first one, $[s_{1},t_{1}] = [1,k]$, starts and ends at the first occurrence of $\C_1$ and $\C_k$, respectively.
    \item For each $i \in [2,k]$, the minimal co-occurrence $[s_{i},t_{i}]$ starts at the first occurrence of $\C_i$ and ends at the second occurrence of $\C_{i-1}$, i.e., $[s_i,t_i] = [i, i-1 + e_{i-1}]$. 
\end{itemize} 

Consider how these minimal co-occurrences contribute to $\delta$. Each $[s_{i},t_{i}]$ contributes plus one to $\delta(\len(s_{i},t_{i}))$. For the first co-occurrence, $\len(s_{1},t_{1}) = k$ (which is not a part of the universe $V$). Each of the other co-occurrences $[s_{i},t_{i}]$ has length
\[
    \len(s_{i},t_{i}) \;=\; t_{i} - s_{i} + 1 \;=\;  (i-1 + e_{i-1}) - i + 1 \;=\; e_{i-1}
\] Therefore, the remaining minimal co-occurrences contribute plus one to of each the $\delta(\cdot)$-entries $e_1,\ldots,e_{k-1}$.

Furthermore, each $[s_{i},t_{i}]$ contributes negative one to $\delta(\len(s_{i},t_{i+1}))$, where we define $t_{k+1} = |R_j| + 1$. For $i < k$ we have that $\len(s_{i},t_{i+1}) = 1 + \len(s_{i+1},t_{i+1}) = 1 + e_i$ since $s_{i} + 1 = s_{i+1}$. Note that $1 + e_i$ is odd since $e_i$ is even, and therefore not in $V$. For $i = k$, we get that $t_{i+1} = t_{k+1} = |R_j| + 1$. The last $k\alpha$ (at least) characters of $R_j$ are not in $Q$, so $\len(s_{k},|R_j| + 1) > k\alpha$ and therefore not in $V$.   

Hence, $R_j$ contributes plus one to $\delta(e_1),\ldots,\delta(e_{k-1})$ and does not contribute anything to $\delta(i)$ for any other $i \in V\setminus T_j$. To construct $S$, concatenate $R_1,\ldots,R_t$. Note that any minimal co-occurrence that crosses the boundary between two gadgets will only contribute to $\delta(i)$ for $i > k\alpha$ due to the trailing characters of each gadget that are not in $Q$. Since $S$ consists of $t = O(\alpha)$ gadgets that each have length $O(k\alpha)$, we have $n = O(k\alpha^2)$ as stated above.

Finally, note that the assumption that $|T|$ is a multiple of $k-1$ is not necessary. We ensure that the size is a multiple of $k-1$ by adding at most $k-2$ \emph{even} integers from $\{k\alpha+1,\ldots,2k\alpha\}$ and adjusting the size of the gadgets accordingly. The reduction still works because we add even integers, the size of $S$ is asymptotically unchanged, and any minimal co-occurrence due to the extra elements will have length greater than $k\alpha$ and will not contribute to any relevant $\delta$-entries.

\paragraph*{Acknowledgements} We would like to thank the anonymous reviewers for their comments, which improved the presentation of the paper. 

\appendix

\section{Preprocessing}\label{app:preprocessing}
\paragraph{Finding Minimal Co-Occurrences}
To build the data structure, we need to find all the minimal co-occurrences in order to determine $\delta$. For $j \geq \rep{1}$, let $\lm(j)$ denote the length of the left-minimal co-occurrence ending at index $j$. By Lemma~\ref{lem:lmm_start}, $\lm(\rep{i}) = \len(\lep{i},\rep{i})$ for each $i \in [1,\m]$. Furthermore, for $j \in [\rep{i}+1,\rep{i+1}-1]$ we have $\lm(j) = \lm(j-1) + 1$ since both of the left-minimal co-occurrences ending at these two indices start at $\lep{i}$. However, $\lm(\rep{i}) \leq \lm(\rep{i}-1)$ for each $i \in [2,\m]$; the left-minimal co-occurrence ending at $\rep{i}$ starts at least one index further to the right than the left-minimal co-occurrence ending at $\rep{i}-1$ because $\lep{i-1} < \lep{i}$, so it cannot be strictly longer.

We determine $\lep{1},\ldots,\lep{\m}$ and $\rep{1},\ldots,\rep{\m}$ using the following algorithm. Traverse $S$ and maintain $\lm(j)$ for the current index $j$. Whenever $\lm(j) \neq \lm(j-1) + 1$ the interval $[j - \lm(j) + 1,\; j]$ is one of the minimal co-occurrences. Note that this algorithm finds the minimal co-occurrences in order by their rightmost endpoint. We maintain $\lm(j)$ as follows. For each character $p \in Q$ let $\textsf{dist}(j,p)$ be the distance to the closest occurrence of $p$ on the left of $j$. Then $\lm(j)$ is the maximum $\textsf{dist}(j,\cdot)$-value. As in~\cite{SBDN+2021}, we maintain the $\textsf{dist}(j,\cdot)$-values in a linked list that is dynamically reordered according to the well-known \emph{move-to-front} rule. The algorithm works as follows. Maintain a linked list over the elements in $Q$, ordered by increasing $\textsf{dist}$-values. Whenever you see some $p \in Q$, access its node in expected constant time through a dictionary and move it to the front of the list. The least recently seen $p \in Q$ (i.e., the $p$ with the largest $\textsf{dist}(j,\cdot)$-value) is found at the back of the list in constant time. The algorithm uses $O(q)$ space and expected constant time per character in $S$, thus it runs in expected $O(n)$ time. 

\paragraph{Building the Data Structure}
We build the data structure as follows. Traverse $S$ and maintain the two most recently seen minimal co-occurrences using the algorithm above. We maintain the non-zero $\delta(\cdot)$-values in a dictionary $D$ that is implemented using chained hashing in conjunction with universal hashing~\cite{CW1979}. When we find a new minimal co-occurrence $[\lep{i+1},\rep{i+1}]$ we increment $D[\len(\lep{i},\rep{i})]$ and decrement $D[\len(\lep{i},\rep{i+1})]$. Recall that $Z = \{z_1,\ldots,z_d\}$ where $z_j < z_{j+1}$ is defined such that $\delta(i) \neq 0$ if and only if $i \in Z$. After processing $S$ the dictionary $D$ encodes the set $\{(z_1,\delta(z_1)),\ldots,(z_d,\delta(z_d))\}$. Sort the set to obtain the array $E[j] = \delta(z_j)$. Compute the partial sum array over $E$, i.e the array 
    \[
        F[j] = \sum_{i=1}^j E[i] = \sum_{i=1}^j \delta(z_i) = \lmm(z_j). \qquad\qquad\qquad \text{(we use $1$-indexing)}
    \]
Build the predecessor data structure over $Z$ and associate $\lmm(z_j)$ with each key $z_j$.

The algorithm for finding the minimal co-occurrences uses $O(q)$ space and the remaining data structures all use $O(d)$ space, for a total of $O(d + q)$ space. Finding the minimal co-occurrences and  maintaining $D$ takes $O(n)$ expected time, and so does building the predecessor structure from the sorted input.

Furthermore, we use the following sorting algorithm to sort the $d$ entries in $D$ with $O(d)$ extra space in expected $O(n)$ time. If $d < n/\log n$ we use merge sort which uses $O(d)$ extra space and runs in $O(d\log d) = O(n)$ time. If $d \geq n/\log n$ we use radix sort with base $\sqrt{n}$, which uses $O(\sqrt{n})$ extra space and $O(n)$ time. To elaborate, assume without loss of generality that $2k$ bits are necessary to represent $n$. We first distribute the elements into $2^k = O(\sqrt{n})$ buckets according to the most significant $k$ bits of their binary representation, partially sorting the input. We then sort each bucket by distributing the elements in that bucket according to the \emph{least} significant $k$ bits of their binary representation, fully sorting the input. The algorithm runs in $O(n)$ time and uses $O(\sqrt{n}) = O(n/\log n) = O(d)$ extra space. 

\section{Lower Bound on Time}\label{app:lower_bound_on_time}
We now prove Theorem~\ref{thm:time_space_tradeoff} by the following reduction from the predecessor problem. Let $U$, $Q$ and $G_i$ be as defined in Section~\ref{sec:increment_gadget} and let $X = \{x_1,x_2,\ldots,x_m\} \subseteq U$ where $x_1 < \ldots < x_m$. Define
\[
    S = \underbrace{G_{x_1}\cdots G_{x_1}}_{x_1}\s \underbrace{G_{x_2}\cdots G_{x_2}}_{x_2 - x_1} \s\ldots\ldots\ldots\s \underbrace{G_{x_m}\cdots G_{x_m}}_{x_m - x_{m-1}}
\] 
By Lemma~\ref{lem:increment_gadget} we have that $\delta(x_1) = x_1$, $\delta(x_i) = x_i - x_{i-1}$ for $i \in [2,m]$ and $\delta(i) = 0$ for $i \in U\setminus X$. Then, if the predecessor of some $x \in U$ is $x_p$, we have 
\[
    \lmm(x) = \sum_{i=2}^x \delta(i) = x_1 + (x_2 - x_1) + \ldots + (x_p - x_{p-1}) = x_p    
\]
On the other hand, if  $x < x_1$ then  $\sum_{i=0}^x\delta(i) = 0$, unambiguously identifying that $x$ has no predecessor.

Applying Lemma~\ref{lem:increment_gadget} again, we have $d \leq 8m$. Furthermore, there are $x_1 + (x_2 - x_1) + \ldots + (x_m - x_{m-1}) = x_m \leq u$ gadgets in total so $n \leq 2u^2$. Hence, given a data structure that supports $\lmm$ in $f_t(n,q,d)$ time using $f_s(n,q,d)$ space, we get a data structure supporting predecessor queries on $X$ in $O(f_t(2u^2, 2, 8m))$ time and $O(f_s(2u^2,2,8m))$ space, proving Theorem~\ref{thm:time_space_tradeoff}.

\end{document}